\documentclass[11pt,letterpaper]{amsart}

\usepackage{amsmath,amsfonts,amsthm,amssymb,stmaryrd,bm,cite,enumerate}
\theoremstyle{plain}

\numberwithin{equation}{section}

\newtheorem{thm}{Theorem}[section]
\newtheorem{cor}[thm]{Corollary}

\theoremstyle{definition}
\newtheorem{example}{Example}

\allowdisplaybreaks  

\newcommand{\complex}{{\mathbb C}}
\newcommand{\real}{{\mathbb R}}

\newcommand{\tbullet}{\mathrel{\raise .2ex\hbox{\tiny$\bullet$}}} 

\newcommand{\trace}{\mathrm{tr\,}}
\newcommand{\rmob}{\mathrm{Ob\,}}
\newcommand{\ityes}{\textit{yes}}      
\newcommand{\itno}{\textit{no}}
 \newcommand{\Var}{\operatorname{Var}}
 \newcommand{\rmre}{\mathrm{Re}}

\newcommand{\escript}{\mathcal{E}}
\newcommand{\lscript}{\mathcal{L}}
\newcommand{\pscript}{\mathcal{P}}
\newcommand{\sscript}{\mathcal{S}}
\newcommand{\hhat}{\widehat{h}}

\newcommand{\ab}[1]{\left|#1\right|}
\newcommand{\brac}[1]{\left\{#1\right\}}
\newcommand{\paren}[1]{\left(#1\right)}
\newcommand{\sqbrac}[1]{\left[#1\right]}

\newcommand{\doubleab}[1]{\left|\left|#1\right|\right|}
\newcommand{\elbows}[1]{{\left\langle#1\right\rangle}}
\newcommand{\ket}[1]{{\left|#1\right>}}
\newcommand{\bra}[1]{{\left<#1\right|}}

\begin{document}

\title{UNCERTAINTY OF QUANTUM STATES}
\author{Stan Gudder}
\address{Department of Mathematics, 
University of Denver, Denver, Colorado 80208}
\email{sgudder@du.edu}
\date{}
\maketitle

\begin{abstract}
The uncertainty of a quantum state is given by the composition of two components. The first is called the quantum component and is given by the probability distribution of an observable relative to the state. The second is the classical component which is an uncertainty function that is applied to the first component. We characterize uncertainty functions in terms of four axioms. We then study four examples called variance, entropy, geometric and sine uncertainty functions. The final section presents the general theory of state uncertainty.
\end{abstract}

\section{Introduction}  
Uncertainty in quantum mechanics was first introduced by Heisenberg to describe restrictions on the accuracy of measurement results for two or more noncommuting observables \cite{hei72}. These uncertainty relations were later systemized by various authors \cite{cvtw17,deu83,fgg13,mu88,rob29}. They have become the central ingredient in the security analysis of many cryptographic protocols such as quantum key distribution \cite{tr11}. They have also been used as a powerful tool for entanglement witness, EPR steering \cite{sbwch13} and quantum metrology \cite{glm11}. These studies have been useful for describing the uncertainty of quantum measurements. But at a more fundamental level, what is the uncertainty of a quantum state and what does this actually mean? One can safely say that the identification of the state of a quantum system is one of the most important problems in quantum mechanics. It is certainly of great importance in quantum computation and information theory \cite{hz12,nc00}. If we have a numerical measure of the uncertainty or equivalently the certainty of a quantum state then we can choose the most likely from a finite set of possible choices. In this article, we point out that the answer to the above question is the reverse of the measurement uncertainty procedure. If $\rho$ is an unknown quantum state, the uncertainty of $\rho$ depends on the context (observable) under which the system is being measured. We shall follow the work in
\cite{zt24}. However, in this previous article, only atomic projection-valued measures are considered and we shall employ the more general positive operator-valued measures (POVM). Moreover, our approach is more complete.

In this article, we shall only consider quantum systems described by a finite dimensional Hilbert space $H$. In this case, an \textit{observable} is a finite set $A=\brac{A_0,A_1,\ldots ,A_{d-1}}$ where $A_j\ne 0$ is a positive operator, $i=0,1,\ldots ,d-1$ with
$\sum\limits _{i=0}^{d-1}A_i=I$ and $I$ is the identity operator. We call $i=0,1,\ldots ,d-1$ the \textit{outcomes} of $A$ and according to Born's rule, the probability that outcome $i$ is observed when $A$ is measured and the system is in the state $\rho$ is the trace $\trace (\rho A_i)$. In this introduction, we shall only present rough ideas and a more detailed study will be given in Section~2. In the context of the observable $A$, let $U_A(\rho )$ be the uncertainty of the state $\rho$ where we assume that $0\le U_A(\rho )\le 1$. The most important requirement of $U_A$ is \textit{concavity} which postulates that $U_A$ cannot decrease under mixtures. Mathematically this means that if
$\lambda\rho _1+(1-\lambda )\rho _2$ where $0\le\lambda\le 1$, is a mixture of the states $\rho _1,\rho _2$ then
\begin{equation*}
U_A\sqbrac{\lambda\rho _1+(1-\lambda )\rho _2}\ge\lambda U_A(\rho _1)+(1-\lambda )U_A(\rho _2)
\end{equation*}
Now there are two special types of states. One is an $A$-\textit{maximal uncertain} state $\rho$ that satisfies $\trace (\rho A_i)=\tfrac{1}{d}$,
$i=0,1,\ldots ,d-1$. For example, if $A$ is the projection-valued observable $A=\brac{\ket{\phi _i}\bra{\phi _i}\colon i=0,1,\ldots ,d-1}$ where
$\brac{\phi _i}$ is an orthonormal basis for $H$, then the completely random state $\tfrac{1}{d}\,I$ is $A$-maximal uncertain. We require that if $\rho$ is $A$-maximal uncertain, then $U_A(\rho )$ has its largest value, namely $U_A(\rho )=1$. The second special type of state is an $A$-\textit{maximal certain} state $\rho$ that satisfies $\trace (\rho A_i)=1$ for some $i\in\brac{0,1,\ldots ,d-1}$. For example, if again
$A=\brac{\ket{\phi _i}\bra{\phi _i}\colon i=0,1,\ldots ,d-1}$, then $\ket{\phi _i}\bra{\phi _i}$ is an $A$-maximal certain state for all $j\in\brac{0,1,\ldots ,d-1}$. We require that if $\rho$ is $A$-maximal certain, then $U_A(\rho )$ has its smallest value, namely $U_A(\rho )=0$. We will also have a fourth symmetric requirement which is not very important.

We shall see in Sections~2 and 4 that this is not the total picture. We have only discussed the quantum component of state uncertainty. There is also a classical component given by an uncertainty function on the set of probability distributions and is discussed in Section~2
\cite{luo05,ls17}. In Section~3, we study four examples of uncertainty functions. We call these variance, entropy, geometric and sine uncertainty functions. Section~4 presents the general theory of state uncertainty.

\section{Uncertainty Functions}  
In this article, $H$ will denote a complex, finite-dimensional Hilbert space. The set of self-adjoint operators on $H$ is denoted by
$\lscript _S(H)$ and a positive operator $\rho\in\lscript _S(H)$ with trace $\trace (\rho )=1$ is called a \textit{density operator} or \textit{state}. We denote the set of states on $H$ by $\sscript (H)$. A state corresponds to an initial condition for a quantum system. An operator $a\in\lscript _S(H)$ satisfying $0\le a\le I$ is called an \textit{effect} and the set of effects is denoted by $\escript (H)$. An effect $a$ corresponds to a measurement that has two outcomes \ityes --\itno\ (true-false). If a measurement of $a$ results in outcome \ityes\ we say that $a$
\textit{occurs} and if it results in outcome \itno\ we say that $a$ \textit{does not occur}. The effect $I$ always occurs and the effect $0$ never occurs. If $a\in\escript (H)$, then its \textit{complement} is $a'=I-a\in\escript (H)$ and $a$ occurs under a specific measurement if and only if $a'$ does not occur under that measurement. If $\rho\in\sscript (H)$ and $a\in\escript (H)$, then the probability that $a$ occurs when the system is in state $\rho$ is given by $P_\rho (a)=\trace (\rho a)$. Notice that 
\begin{equation*}
P_\rho (a)+P_\rho (a')=\trace (\rho a)+\trace (\rho a')=1
\end{equation*}
More generally, if $a,b\in\escript (H)$ with $a+b\in\escript (H)$ we write $a\perp b$ and in this case
\begin{equation*}
P_\rho (a)+P_\rho (b)=\trace (\rho a)+\trace (\rho b)=\trace\sqbrac{\rho (a+b)}\le 1
\end{equation*}

An \textit{observable} $A$ for a quantum system is given by a finite collection of effects $A=\brac{A_i\colon i=0,1,\ldots ,d-1}$ where the $A_i\in\escript (H)$, $A_i\ne 0$, need not be distinct and $\sum\limits _{i=0}^{d-1}A_i=I$. We call $\brac{0,1,\ldots ,d-1}$ the
\textit{outcome space} for $A$ and if a measurement of $A$ results in outcome $i$, we say that $A_i$ \textit{occurs}. We denote the set of observables on $H$ by $\rmob (H)$. If $\rho\in\sscript (H)$ and $A\in\rmob (H)$, then the probability that outcome $i$ occurs where $A$ is measured and the system is in state $\rho$ is given by $P_\rho (A_i)$. Since 
\begin{equation*}
\sum _iP_\rho (A_i)=\sum _i\trace (\rho A_i)=\trace\paren{\rho\sum _iA_i}=\trace (\rho I)=\trace (\rho )=1
\end{equation*}
we see that $P_\rho ^A(i)=P_\rho (A_i)$ is a probability disribution.

If $d>1$ is an integer, a \textit{probability simplex} is a set of probability distributions
\begin{equation*}
\pscript _d=\brac{(x_0,x_1,\ldots ,x_{d-1})\colon 0\le x_i\le 1,\sum _{i=0}^{d-1}x_i=1}
\end{equation*}
We treat the elements of $\pscript _d$ as real vectors and call $x\in\pscript _d$ a \textit{probability vector}. If $\lambda\in\sqbrac{0,1}$,
$x,y\in\pscript _d$, we see that 
\begin{align*}
\lambda x+(1-\lambda )y&=\lambda (x_0,x_1,\ldots ,x_{d-1})+(1-\lambda )(y_0,y_1,\ldots ,y_{d-1})\\
   &=\sqbrac{\lambda x_0+(1-\lambda )y_0,\lambda x_1+(1-\lambda )y_1,\ldots ,\lambda x_{d-1}+(1-\lambda )y_{d-1}}
\end{align*}
Since $0\le\lambda x_i+(1-\lambda )y_i\le 1$, $i=0,1,\ldots ,d-1$ and
\begin{equation*}
\sum _{i=0}^{d-1}\sqbrac{\lambda x_i+(1-\lambda )y_i}
   =\lambda\sum _{x=0}^{d-1}x_i+(1-\lambda )\sum _{i=0}^{d-1}y_i=\lambda +(1-\lambda )=1
\end{equation*}
we conclude that $\lambda x+(1-\lambda )y\in\pscript _d$ so $\pscript _d$ is a convex set. A function $f\colon\pscript _d\to\sqbrac{0,1}$ is
\textit{concave} if for all $x,y\in\pscript _d$ and $\lambda\in\sqbrac{0,1}$ we have
\begin{equation}                
\label{eq21}
f\sqbrac{\lambda x+(1-\lambda )y}\ge\lambda f(x)+(1-\lambda )f(y)
\end{equation}
If $f$ is concave, the stronger inequality 
\begin{equation}                
\label{eq22}
f\paren{\sum _i\lambda _iz_i}\ge\sum _i\lambda _if(z_i)
\end{equation}
holds for all $z_i\in\pscript _d$, $\lambda _i\in\sqbrac{0,1}$, $i=1,2,\ldots ,n$, with $\sum\limits _{i=1}^n\lambda _i=1$. Indeed, if
$\lambda _1,\lambda _2,\lambda _3\in\sqbrac{0,1}$ with $\lambda _1+\lambda _2+\lambda _3=1$, $\lambda _1\ne 1$, we have
\begin{align*}
f(\lambda _1z_1+\lambda _2z_2+\lambda _3z_3)
   &=f\sqbrac{\lambda _1z_1+(1-\lambda _1)\paren{\tfrac{1}{1-\lambda _1}\lambda _2z_2+\tfrac{1}{1-\lambda _1}\lambda _3z_3}}\\
   &\ge\lambda _1f(z_1)+(1-\lambda _1)f\paren{\tfrac{1}{1-\lambda _1}\lambda _2z_2+\tfrac{1}{1-\lambda _1}\lambda _3z_3}\\
   &\ge\lambda _1f(x_1)+(1-\lambda _1)\sqbrac{\tfrac{1}{1-\lambda _1}\lambda _2f(z_2)+\tfrac{1}{1-\lambda _1}f(z_3)}\\
   &=\lambda _1f(z_1)+\lambda _2f(z_2)+\lambda _3f(z_3)
\end{align*}
If $\lambda _1=1$, the inequality holds trivially. Continuing by induction, we conclude that \eqref{eq22} holds. If
$h\colon\brac{0,1,\ldots ,d-1}\to\brac{0,1,\ldots ,d-1}$ is a bijection, we define $\hhat\colon\pscript _d\to\pscript _d$ by
\begin{equation*}
\hhat\sqbrac{(x_0,x_1,\ldots ,x_{d-1})}=(x_{h(0)},x_{h(1)},\ldots ,x_{h(d-1)})
\end{equation*}
and call $\hhat$ a \textit{permutation} on $\pscript _d$. A function $f\colon\pscript _d\to\sqbrac{0,1}$ is \textit{symmetric} if
$f\sqbrac{\hhat (x)}=f(x)$ for any permutation $\hhat$ and $x\in\pscript _d$.

We call an $x\in\pscript _d$ satisfying $x_i=1$ for some $i\in\brac{0,1,\ldots ,d-1}$ a \textit{maximal certainty distribution} and if
$x=\paren{\tfrac{1}{d},\tfrac{1}{d},\cdots ,\tfrac{1}{d}}$ we call $x$ a \textit{maximal uncertainty distribution}. Notice that there is only one maximal uncertainty distribution and only one maximal certainty distribution to within a permutation. We now present our main definition
\cite{mu88,sbwch13,tr11}. An \textit{uncertainty function on} $\pscript _d$ is a map $f\colon\pscript _d\to\sqbrac{0,1}$ that satiisfies:
\begin{align}                
\label{eq23}
&f(x)=0\hbox{ if and only if }x\hbox{ is a maximal certainty distribution,}\\
\label{eq24}
&f(x)=1\hbox{ if and only if }x\hbox{ is a maximal uncertainty distribution,}\\
\label{eq25}
&f\hbox{ is symmetric},\\
\label{eq26}
&f\hbox{ is concave}.
\end{align}

The following theorem is useful.
\begin{thm}    
\label{thm21}
Let $f\colon\pscript _d\to\real$ be given by $f(x)=\sum\limits _{i=0}^{d-1}h(x_i)$ where $h\colon\sqbrac{0,1}\to\sqbrac{0,1}$. Then $f$ is an uncertainty function if
{\rm{(i)}}\enspace $h$ is concave,
{\rm{(ii)}}\enspace $h(0)=h(1)=0$,
{\rm{(iii)}}\enspace $h(\alpha )>0$ if $\alpha\ne 0,1$,
{\rm{(iv)}}\enspace $h\paren{\tfrac{1}{d}}=\tfrac{1}{d}$,
{\rm{(v)}}\enspace $\sum\limits _{i=0}^{d-1}h(x_i)=1$ implies $x_i=\tfrac{1}{d}$, $i=0,1,\ldots ,d-1$.
\end{thm}
\begin{proof}
It is clear that $f$ is symmetric. Also, $f$ is concave because by (i)
\begin{align*}
f\sqbrac{\lambda x+(1-\lambda )y}&=\sum _{i=0}^{d-1}h\sqbrac{\lambda x_i+(1-\lambda )y_i}
   \ge\sum _{i=0}^{d-1}\sqbrac{\lambda h(x_i)+(1-\lambda )h(y_i)}\\
   &=\lambda\sum _{i=0}^{d-1}h(x_i)+(1-\lambda )\sum _{i=0}^{d-1}h (y_i)=\lambda f(x)+(1-\lambda )f(y)
\end{align*}
To check \eqref{eq23}, suppose $x_i=1$ for some $i$. Then by (ii) $f(x)=f(1)=0$. Conversely, if $f(x)=0$, then
$\sum\limits _{i=0}^{d-1}h(x_i)=0$. Hence, $h(x_i)=0$, $i=0,1,\ldots ,d-1$ so by (ii) and (iii), $x_i=0$ or $1$, $i=0,1,\ldots ,d-1$. We conclude that $x_i=1$ for some $i$. To check \eqref{eq24}, suppose $y=\paren{\tfrac{1}{d},\tfrac{1}{d},\cdots ,\tfrac{1}{d}}$. Then by (iv)
\begin{equation*}
f(y)=\sum _{i=0}^{d-1}h(y_i)=\sum _{i=0}^{d-1}h\paren{\tfrac{1}{d}}=df\paren{\tfrac{1}{d}}=1
\end{equation*}
Conversely, if $f(y)=1$, then $\sum\limits _{i=0}^{d-1}h(y_i)=1$, so by (v), $y_i=\tfrac{1}{d}$. Finally, $f(x)\le 1$ for all $x\in\pscript _d$ because (i) implies
\begin{equation*}
h\paren{\sum _{i=0}^{d-1}\lambda _ix_i}\ge\sum _{i=0}^{d-1}\lambda _ih(x_i)
\end{equation*}
for all $\lambda _i\in\sqbrac{0,1}$ with $\sum\lambda _i=1$. Letting $\lambda _i=\tfrac{1}{d}$, $i=0,1,\ldots ,d-1$ gives
\begin{equation*}
\tfrac{1}{d}\sum h(x_i)\le h\paren{\tfrac{1}{d}\sum x_i}=h\paren{\tfrac{1}{d}}
\end{equation*}
Hence, by (iv) we obtain
\begin{equation*}
f(x)=\sum _{i=0}^{d-1}h(x_i)\le dh\paren{\tfrac{1}{d}}=1\qedhere
\end{equation*}
\end{proof}

\section{Examples of Uncertainty Functions}  
We now list some uncertainty functions. The first three are studied except for a multiple factor \cite{zt24} while the fourth appears to be new. We denote the natural logarithm function by $\ln$.
\begin{enumerate}[(a)]
\item \textit{Variance uncertainty}: $v(x)=\tfrac{d}{d-1}\paren{1-\sum\limits _{i=0}^{d-1}x_i^2}$
\item \textit{Entropy uncertainty}: $e(x)=-\tfrac{1}{\ln d}\sum\limits _{i=0}^{d-1} x_i\ln x_I$
\item \textit{Geometric uncertainty}: $g(x)=\tfrac{d}{d-1}(1-\max x_i)$
\item \textit{Sine uncertainty}: $s(x)=\tfrac{1}{d\sin (\pi /d)}\sum\limits _{i=0}^{d-1}\sin (\pi x_i)$
\end{enumerate}
We now show that these are indeed uncertainty functions.

\begin{thm}    
\label{thm31}
The maps $v,e,g$ and $s$ are uncertainty functions.
\end{thm}
\begin{proof}
We can write $v(x)=\sum\limits _{i=0}^{d-1}h(x_i)$, $e(x)=\sum\limits _{i=0}^{d-1}h_1(x_i)$,\newline
$s(x)=\sum\limits _{i=0}^{d-1}h_2(x_i)$ where
$h, h_1,h_2\colon\sqbrac{0,1}\to\sqbrac{0,1}$ are given by\newline
 $h(\alpha )=\tfrac{d}{d-1}(\alpha -\alpha ^2)$, $h_1(\alpha )=-\tfrac{1}{\ln d}\alpha\ln\alpha$,
 $h_2(\alpha )=\tfrac{1}{d\sin (\pi /d)}\sin (\pi\alpha )$.
 We now apply Theorem~\ref{thm21} to show that $v,e$ and $s$ are uncertainty functions. Since the derivatives
$h'(\alpha )=\tfrac{d}{d-1}(1-2d)$, $h''(\alpha )=-\tfrac{2d}{d-1}<0$ we conclude that $h$ is concave and $h$ has its maximum at
$\alpha =1/2$. Since $h\paren{\tfrac{1}{2}}=\tfrac{d}{4(d-1)}$ and $d>\tfrac{4}{3}$ we obtain 
\begin{equation*}
\tfrac{d-1}{d}=1-\tfrac{1}{d}>1-\tfrac{3}{4}=\tfrac{1}{4}
\end{equation*}
so $h(\alpha )\le h\paren{\tfrac{1}{2}}\le 1$ as required. Clearly, (ii) and (iii) of Theorem~\ref{thm21} hold for $h$. Moreover,
\begin{equation*}
h\paren{\tfrac{1}{d}}=\tfrac{d}{d-1}\paren{\tfrac{1}{d}-\tfrac{1}{d^2}}=\tfrac{d}{d-1}\paren{\tfrac{d-1}{d^2}}=\tfrac{1}{d}
\end{equation*}
so (iv) holds. To show that (v) holds, suppose
\begin{equation*}
\tfrac{d}{d-1}\paren{1-\sum _{i=0}^{d-1}x_i^2}=\sum _{i=0}^{d-1}h(x_i)=1
\end{equation*}
which implies
\begin{equation*}
\sum _{i=0}^{d-1}x_i^2=1-\tfrac{d-1}{d}=\tfrac{1}{d}
\end{equation*}
Letting $\beta =(1,1,\ldots ,1)$, the inner product with $x$ becomes
\begin{equation*}
\elbows{x,\beta}=\sum _{i=0}^{d-1}x_i\tbullet 1=1=\paren{\sum _{i=0}^{d-1}x_i^2}^{1/2}d^{1/2}=\doubleab{x}\,\doubleab{\beta}
\end{equation*}
Since we have equality in Schwarz's inequality, we conclude that $x=\lambda\beta$ for some $\lambda\in\real$. Hence, $x_i=x_j$ for all $i,j$ and since $\sum\limits _{i=0}^{d-1}x_i=1$, it follows that $x_i=\tfrac{1}{d}$, $i=0,1,\ldots ,d-1$. Therefore, (v) holds.

The function $h_1$ is the well studied entropy function and it is well known that $h_1$ satisfies the conditions of Theorem~\ref{thm21}. Hence, $e$ is an uncertainty function. We next consider $h_2$. Since $h''_2(\alpha )\le 0$ for all $\alpha\in\sqbrac{0,1}$, $h_2$ is concave so (i) holds. It is clear that (ii), (iii) and (iv) hold for $h_2$. To show that (v) holds for $h_2$, suppose $s(x)=\sum\limits _{i-0}^{d-1}h_2(x_i)=1$. As in the proof of Theorem~\ref{thm21}, $s(x)\le 1$ for all $x\in\pscript _d$. Since $s(x)=1$, we have that $s(x)$ is a maximum for $s$. Employing the method of Lagrange multipliers, let
\begin{equation*}
\lscript (x,\lambda )=s(x)+\lambda g(x)
\end{equation*}
for $x\in\real ^d$, where $g(x)=1-\sum\limits _{i=0}^{d-1}x_i$. Then $s(x)$ attains a maximum if
$\tfrac{\partial s(x)}{\partial x_i}+\tfrac{\partial g(x)}{\partial x_i}=0$ for $i=0,1,\ldots ,d-1$. Now
\begin{equation*}
0=\tfrac{\partial s(x)}{\partial x_i}+\tfrac{\lambda\partial g(x)}{\partial x_i}=\tfrac{\pi}{d\sin (\pi /d)}\cos (\pi x_i)-\lambda
\end{equation*}
Hence, $cos (\pi x_i)=\tfrac{\lambda\sin (\pi /d)}{\pi}$ so $\cos (\pi x_i)=\cos (\pi x_j)$ for all $i,j=0,1,\ldots ,d-1$. This implies that $x_i=x_j$ so $x_i=\tfrac{1}{d}$, $i=0,1,\ldots ,d-1$. Hence, (v) holds.

Finally, we show that $g(x)=\tfrac{d}{d-1}(1-\max x_i)$ is an uncertainty function. To show that $g$ is concave, let
\begin{equation*}
\lambda x_j+(1-\lambda )y_j=\max\sqbrac{\lambda x_i+(1-\lambda )y_i}
\end{equation*}
We then have
\begin{align*}
g\sqbrac{\lambda x+(1-\lambda )y}&=\tfrac{d}{d-1}\sqbrac{1-\max (\lambda x_i+(1-\lambda )y_i)}\\
   &=\tfrac{d}{d-1}\sqbrac{1-(\lambda x_j+(1-\lambda )y_j)}
\end{align*}
Since $\max x_i\ge x_j$ and $\max y_i\ge y_j$ we obtain
\begin{align*}
g(x)&=\tfrac{d}{d-1}(1-\max x_i)\le\tfrac{d}{d-1}(1-x_j)\\
\intertext{and}
g(y)&=\tfrac{d}{d-1}(1-\max y_i)\le\tfrac{d}{d-1}(1-y_j)
\end{align*}
Hence, $x_j\le 1-\tfrac{d-1}{d}g(x)$ and $y_j\le 1-\tfrac{d-1}{d}g(y)$. Therefore,
\begin{align*}
\lambda g(x)+(1-\lambda )g(y)&\le\lambda\tfrac{d}{d-1}(1-x_j)+(1-\lambda )\tfrac{d}{d-1}(1-y_j)\\
   &=\tfrac{d}{d-1}\sqbrac{1-(\lambda x_j+(1-\lambda )y_j)}=g(\lambda x+(1-\lambda )y)
\end{align*}
so $g$ is concave. To show that $g$ satisfies \eqref{eq23}, if $x_i=1$ for some $i\in\brac{0,1,\ldots ,d-1}$ then $x_j=0$ for $j\ne i$. Hence, $g(x)=\tfrac{d}{d-1}(1-\max x_i)=0$. Conversely, if $g(x)=0$, then $1-\max x_i=1-x_j=0$. Hence, $x_j=\max x_i=1$ so $x_i=1$ for some $i$. To show that $g$ satisfies \eqref{eq24}, if $y=\paren{\tfrac{1}{d},\tfrac{1}{d},\cdots ,\tfrac{1}{d}}$ then
$g(y)=\tfrac{d}{d-1}\paren{1-\tfrac{1}{d}}=1$. Conversely, if $g(x)=1$, then $\tfrac{d}{d-1}(1-\max x_i)=1$ so $\max x_i=\tfrac{1}{d}$. If
$x_j<\tfrac{1}{d}$ for some $j\in\brac{0,1,\ldots ,d-1}$, then $\sum\limits _{i=0}^{d-1}x_i<1$ which is a contradiction. Hence,
$x=\paren{\tfrac{1}{d},\tfrac{1}{d},\cdots ,\tfrac{1}{d}}$. Clearly, $g$ is symmetric. Finally, since $\sum\limits _{i=0}^{d-1}x_i=1$, we have
$\max x_i\ge\tfrac{1}{d}$. Hence,
\begin{equation*}
1-\max x_i\le 1-\tfrac{1}{d}=\tfrac{d-1}{d}
\end{equation*}
Therefore, $g(x)=\tfrac{d}{d-1} (1-\max x_i)\le 1$.
\end{proof}

As in the proof of Equation~\eqref{eq22}, an uncertainty function $f\colon\pscript _d\to\sqbrac{0,1}$ satisfies
$f\paren{\sum\limits _{i=0}^{d-1}\lambda _iz_i}\ge\sum\limits _{i=0}^{d-1}\lambda _if(z_i)$ where $\lambda _i\in\sqbrac{0,1}$,
$\sum\limits _{i=0}^{d-1}\lambda _i=1$ and $z_i\in\pscript _d$, $i=0,1,\ldots ,d-1$. We now show that we can combine uncertainty functions to obtain new uncertainty functions. For example, if $\lambda\in\sqbrac{0,1}$, then we can form the convex combination
$\lambda v+(1-\lambda )e$. More generally, we have the following theorem.

\begin{thm}    
\label{thm32}
If $f_1,f_2,\ldots ,f_n$ are uncertainty functions on $\pscript _d$ and $0<\lambda _i\le 1$ with $\sum\limits _{i=1}^n\lambda _i=1$, then $f=\sum\limits _{i=1}^n\lambda _if_i$ is an uncertainty function on $\pscript _d$.
\end{thm}
\begin{proof}
First, it is clear that $0\le f(x)\le 1$ for all $x\in\pscript _d$. To show that \eqref{eq23} holds, suppose $f(x)=0$. Then
$\sum\limits _{i=1}^n\lambda _if_i(x)=0$ implies that $f_i(x)=0$ for all $i\in\brac{1,2,\ldots ,n}$. But then $x_j=1$ for some $j\in\brac{1,2,\ldots ,n}$. Conversely, if $x_j=1$ for some $j\in\brac{1,2,\ldots ,n}$ then $x_k=0$ for $k\ne j$ and $f_i(x)=0$ for all $i\in\brac{1,2,\ldots ,n}$. Hence,
\begin{equation*}
f(x)=\sum _{i=1}^n\lambda _if_i(x)=\sum _{i=0}^n\lambda _i\tbullet 0=0
\end{equation*}
To show that \eqref{eq24} holds, suppose $x=\brac{\tfrac{1}{d},\tfrac{1}{d},\cdots ,\tfrac{1}{d}}$. Then $f_i(x)=1$ for all
$i\in\brac{1,2,\ldots ,n}$ so
\begin{equation*}
f(x)=\sum _{i=1}^n\lambda _if_i(x)=\sum _{i=1}^n\lambda _i=1
\end{equation*}
Conversely, suppose $f(x)=1$. Then $\sum\limits _{i=1}^n\lambda _if_i(x)=1$ so $f_i(x)=1$ for all $i\in\brac{1,2,\ldots ,n}$. Hence,
$x=\brac{\tfrac{1}{d},\tfrac{1}{d},\cdots ,\tfrac{1}{d}}$. To show that \eqref{eq25} holds, let $\hhat\colon\pscript _d\to\pscript _d$ be a permutation. Then
\begin{equation*}
f\paren{\hhat (x_i)}=\sum _{i=1}^n\lambda _if_i\paren{\hhat (x)}=\sum _{i=1}^n\lambda _if_i(x)=f(x)
\end{equation*}
Finally, to show that \eqref{eq26} holds, for $\lambda\in\sqbrac{0,1}$ we obtain
\begin{align*}
f\sqbrac{\lambda x+(1-\lambda )y}&=\sum _{i=1}^n\lambda _if_i\sqbrac{\lambda x+(1-\lambda )y}
   \ge\sum _{i=1}^n\lambda _i\sqbrac{\lambda f_i(x)+(1-\lambda )f_i(y)}\\
   &=\lambda\sum _{i=1}^n\lambda _if_i(x)+(1-\lambda )\sum _{i=1}^n\lambda _if_i(y)\\
   &=\lambda f(x)+(1-\lambda )f(y)\qedhere
\end{align*}
\end{proof}

\section{Uncertainty Measures}  
Let $A=\brac{A_i\colon i=0,1,\ldots ,d-1}$ be an observable on $H$. If $A_i$, $i=0,1,\dots ,d-1$, are projections we call $A$ a
\textit{projective observable} and if $A_i$ are rank one projections we call $A$ an \textit{atomic projective observable}. In this latter case we can write $A_i=\ket{i}\bra{i}$ where $\ket{i}$ is an orthonormal basis for $H$. If $A\in\rmob (H)$ and $\rho\in\sscript (H)$ we have the probability distribution $P_\rho ^A(i)=\trace\sqbrac{\rho (A_i)}$, $i=0,1,\ldots ,d-1$, so $P_\rho ^A\in\pscript _d$. If
$f\colon\pscript _d\to\sqbrac{0,1}$ is an uncertainty function, we define the $(f,A)$-\textit{uncertainty measure}
$U_{(f,A)}\colon\sscript (H)\to\sqbrac{0,1}$ by $U_{(f,A)}(\rho )=f(P_\rho ^A)$. For example, if $f=v$ we have
\begin{equation*}
U_{(v,A)}(\rho )=v(P_\rho ^A)=\tfrac{d}{d-1}\sqbrac{1-\sum _{i=0}^{d-1}\paren{P_\rho ^A(i)}^2}
   =\tfrac{d}{d-1}\sqbrac{1-\sum _{i=0}^{d-1}\paren{\trace (\rho A_i)}^2}
\end{equation*}
In a similar way
\begin{align*}
U_{(e,A)}(\rho )&=e(P_\rho ^A)=-\tfrac{1}{\ln d}\sum _{i=0}^{d-1}\trace (\rho A_i)\ln\paren{\trace (\rho A_i)}\\
U_{(g,A)}(\rho )&=g(P_\rho ^A)=\tfrac{d}{d-1}\brac{1-\max\sqbrac{\trace (\rho A_i)}}\\
U_{(s,A)}(\rho )&=s(P_\rho ^A)=\tfrac{1}{d\sin (\pi /d)}\sum _{i=0}^{d-1}\sin\sqbrac{\pi\trace (\rho A_i)}
\end{align*}

We now show why $U_{(v,A)}$ is called a variance uncertainty measure. If $a\in\escript (H)$, $\rho\in\sscript (H)$, the $\rho$-\textit{variance} of $a$ is the average variation $\Var (a,\rho )=\trace\sqbrac{\rho (a-\trace (\rho a)I)^2}$ from $a$ to its $\rho$-expectation $\trace (\rho a)I$. We write $V(a,\rho )$ in a simpler form
\begin{align*}
\Var (a,\rho )&=\trace\brac{\rho\paren{a^2-2a\trace (\rho a)+\sqbrac{\trace (\rho a)}^2I}}\\
   &=\trace (\rho a^2)-\sqbrac{\trace (\rho a)}^2
\end{align*}
If $A$ is a projective observable, then
\begin{equation*}
\Var (A_i,\rho )=\trace (\rho A_i)-\sqbrac{\trace (\rho A_i)}^2
\end{equation*}
We then have
\begin{align*}
U_{(v,A)}(\rho )&=\tfrac{d}{d-1}\brac{\sum _{i=0}^{d-1}\trace (\rho A_i)-\sum _{i=0}^{d-1}\sqbrac{\trace (\rho A_i)}^2}\\
   &=\tfrac{d}{d-1}\brac{\sum _{i=0}^{d-1}\sqbrac{\trace (\rho A_i)-\paren{\trace (\rho A_i)}^2}}
   =\tfrac{d}{d-1}\sqbrac{\sum _{i=0}^{d-1}\Var (A_i,\rho )}
\end{align*}
which writes $U_{(v,A)}(\rho )$ in terms of the $\rho$-variances of $A_i$.

The next result gives the properties of an uncertainty measure.

\begin{thm}    
\label{thm41}
If $U_{(f,A)}$ is an uncertainty measure, then $0\le U_{(f,A)}(\rho )\le 1$ and we have
{\rm{(i)}}\enspace $U_{(f,A)}\sqbrac{\lambda\rho _1+(1-\lambda )\rho _2}\ge\lambda U_{(f,A)}(\rho _1)+(1-\lambda )U_{(f,A)}(\rho _2)$ for all
$\lambda\in\sqbrac{0,1}$,\newline
{\rm{(ii)}}\enspace $U_{(f,A)}(\rho )=0$ if and only if $\trace (\rho A_i)=1$ for some $i\in\brac{0,1,\ldots ,d-1}$,\newline
{\rm{(iii)}}\enspace $U_{(f,A)}(\rho )=1$ if and only if $\trace (\rho A_i)=\tfrac{1}{d}$ for all $i=\brac{0,1,\ldots ,d-1}$.
\end{thm}
\begin{proof}
It is clear that $0\le U_{(f,A)}(\rho )\le 1$. To prove (i), let
\begin{align*}
x&=\sqbrac{\trace (\rho _1A_0),\trace (\rho _1A_1),\ldots ,\trace (\rho _1A_{d-1})}\\
y&=\sqbrac{\trace (\rho _2A_0),\trace (\rho _2A_1),\ldots ,\trace (\rho _2A_{d-1})}
\end{align*}
Since $f$ is concave we obtain
\begin{align*}
&U_{(f,A)}\sqbrac{\lambda\rho _1+(1-\lambda )\rho _2}\\
   &=f\brac{\trace\sqbrac{(\lambda\rho _1+(1-\lambda )\rho _2)(A_0)},\ldots,\trace\sqbrac{(\lambda\rho _1+(1-\lambda )\rho _2)(A_{d-1})}}\\
   &=f\brac{\sqbrac{\lambda\trace (\rho _1A_0)+(1-\lambda )\trace (\rho _2A_0)},
   \ldots ,\sqbrac{\lambda\trace (\rho _1A_{d-1})+(1-\lambda )\trace(\rho _2A_{d-1})}}\\
   &=f\sqbrac{\lambda x+(1-\lambda )y}\ge\lambda f(x)+(1-\lambda )f(y)=\lambda U_{(f,A)}(\rho _1)+(1-\lambda )U_{(f,A)}(\rho _2)
\end{align*}
for all $\lambda\in\sqbrac{0,1}$. To prove (ii), applying \eqref{eq23} we obtain $U_{(f,A)}(\rho )=0$ if and only if
\begin{equation*}
f\sqbrac{\trace (\rho A_0),\trace (\rho A_1),\ldots ,\trace (\rho A_{d-1})}=0
\end{equation*}
which holds if and only if $\trace (\rho A_i)=1$ for some $i\in\brac{0,1,\ldots ,d-1}$. To prove (iii), applying \eqref{eq24} we obtain
$U_{(f,A)}(\rho )=1$ if and only if
\begin{equation*}
f\sqbrac{\trace (\rho A_0),\trace (\rho A_1),\ldots ,\trace (\rho A_{d-1})}=1
\end{equation*}
which holds if and only if $\trace (\rho A_i)=\tfrac{1}{d}$ for all $i\in\brac{0,1,\dots ,d-1}$.
\end{proof}

As with an earlier argument, it follows from (i) that
\begin{equation*}
U_{(f,A)}\paren{\sum _{i=1}^n\lambda _i\rho _i}\ge\sum _{i=1}^n\lambda _iU_{(f,A)}(\rho _i)
\end{equation*}
for $\lambda _i\in\sqbrac{0,1}$ with $\sum\limits _{i=1}^n\lambda _i=1$.

In \cite{zt24}, the measuring observable is always the atomic, projective observable $A=\brac{\ket{i}\bra{i}\colon i=0,1,\ldots ,d-1}$. However, we contend that this is too restrictive and other observables should also be used. For example, let $A=\brac{A_0,A_1}$ with $A_1=I-A_0$ be an arbitrary observable on the qubit Hilbert space $H=\complex ^2$. Then for any $\rho\in\sscript (H)$ we have the probability distribution
$P_\rho ^A(i)=\trace (\rho A_i)$ $i=0,1$. Consider the mixed state $\rho =\tfrac{1}{2}\paren{\ket{0}\bra{0}+\ket{1}\bra{1}}$ and the pure state
$\ket{\psi}\bra{\psi}$ where $\psi =\tfrac{1}{\sqrt{2}}(\ket{0}+\ket{1})$. We then have for $i=0,1$ that
\begin{align*}
P_\rho ^A(i)&=\trace (\rho A_i)=\tfrac{1}{2}\trace\sqbrac{\paren{\ket{0}\bra{0}+\ket{1}\bra{1}}A_i}\\
   &=\tfrac{1}{2}\sqbrac{\bra{0}A_i\ket{0}+\bra{1}A_i\ket{1}}=\tfrac{1}{2}\trace (A_i)\\
\intertext{and}
P_{\ket{\psi}\bra{\psi}}^A(i)&=\trace\paren{\ket{\psi}\bra{\psi}A_i}=\tfrac{1}{2}\sqbrac{\paren{\ket{0}+\ket{1}}\paren{\ket{0}+\ket{1}}A_i}\\
   &=\tfrac{1}{2}\sqbrac{\bra{0}A_i\ket{0}+\bra{1}A_i\ket{1}+2\rmre\bra{0}A_i\ket{1}}\\
   &=\tfrac{1}{2}\trace (A_i)+\rmre\bra{0}A_i\ket{i}
\end{align*}
We then obtain
\begin{align*}
P_\rho ^A(0)&=\tfrac{1}{2}\trace (A_0),\quad P_\rho ^A(1)=1-\tfrac{1}{2}\trace (A_0)\\
P_{\ket{\psi}\bra{\psi}}^A(0)&=\tfrac{1}{2}\trace (A_0)+\rmre\bra{0}A_0\ket{1},\\
P_{\ket{\psi}\bra{\psi}}^A(1)&=1-\tfrac{1}{2}\trace (A_0)-\rmre\bra{0}A_0\ket{1}
\end{align*}
If $A=\brac{\ket{0}\bra{0},\ket{1}\bra{1}}$ is an atomic, projection observable, then
\begin{equation*}
P_\rho ^A(0)=P_\rho ^A(1)=P_{\ket{\psi}\bra{\psi}}^A(0)=P_{\ket{\psi}\bra{\psi}}^A(1)=1/2
\end{equation*}
so $A$ does not distinguish the two different states $\rho$ and $\ket{\psi}\bra{\psi}$. Also, if $A_0$ is diagonal, then
$\bra{0}A_0\ket{1}=0$, so $A$ does not distinguish $\rho$ and $\ket{\psi}\bra{\psi}$. However, if $A$ satisfies $\rmre\bra{0}A\ket{1}\ne 0$, it does distinguish the states $\rho$ and $\ket{\psi}\bra{\psi}$.

\begin{example}  
Let $H=\complex ^2$ be the qubit Hilbert space and let $\rho$, $\ket{\psi}\bra{\psi}$ be the states just considered. Define the projective observable $A=\brac{A_0,A_1}$ by $A_0=\tfrac{1}{2}\left[\begin{smallmatrix}\smallskip 1&1\\1&1\end{smallmatrix}\right]$,
$A_1=\tfrac{1}{2}\left[\begin{smallmatrix}\smallskip 1&-1\\-1&1\end{smallmatrix}\right]$. Then $\rmre\bra{0}A_0\ket{1}=1/2$ and we have
\begin{equation*}
P_\rho ^A(0)=\tfrac{1}{2}\trace (A_0)=1/2=P_\rho ^A(1)
\end{equation*}
Thus, $P_\rho ^A$ is a maximal uncertainty distribution. We also have 
\begin{equation*}
P_{\ket{\psi}\bra{\psi}}^A(0)=\tfrac{1}{2}\trace (A_0)+\rmre\bra{0}A_0\ket{1}=1
\end{equation*}
and $P_{\ket{\psi}\bra{\psi}}(1)=0$ so $P_{\ket{\psi}\bra{\psi}}^A$ is a maximal certainty distribution, We conclude that $P_\rho ^A$ and
$P_{\ket{\psi}\bra{\psi}}^A$ are exact opposites.\hfill\qedsymbol
\end{example}

\begin{example}  
In Example~1, let $A_0=\tfrac{1}{2}\left[\begin{smallmatrix}\smallskip 1&1/3\\1/3&1\end{smallmatrix}\right]$,
$A_1=\tfrac{1}{2}\left[\begin{smallmatrix}\smallskip 1&-1/3\\-1/3&1\end{smallmatrix}\right]$.
Then $\rmre\bra{0}A_0\ket{1}=1/6$ and we have
\begin{equation*}
P_\rho ^A(0)=P_\rho ^A(1)=1/2
\end{equation*}
as before. Moreover,
\begin{equation*}
P_{\ket{\psi}\bra{\psi}}^A(0)=\tfrac{1}{2}+\tfrac{1}{6}=\tfrac{2}{3},\quad P_{\ket{\psi}\bra{\psi}}^A(1)=\tfrac{1}{3}\hskip 13pc\square
\end{equation*}
\end{example}

\begin{example} 
We now apply our standard uncertainty measures to the previous two examples. In Example~1, we let $x=(1/2,1/2)$, $y=(1,0)$ and obtain $v(x)=e(x)=g(x)=s(x)=1$ and $v(y)=e(y)=g(y)=s(y)=0$. In Example~2, we let $x=(1/2,1/2)$, $z=(2/3,1/3)$. Again, we have
$v(x)=e(x)=g(x)=s(x)=1$. However,
\begin{align*}
v(z)&=2(1-z_0^2-z_1^2)=2\paren{1-\tfrac{4}{9}-\tfrac{1}{9}}=\tfrac{8}{9}=0.889\\
e(z)&=\tfrac{1}{\ln 2}(z_0\ln z_0+z_1\ln z_1)=-\tfrac{1}{\ln 2}\paren{\tfrac{2}{3}\ln\tfrac{2}{3}+\tfrac{1}{3}\ln\tfrac{1}{3}}\\
   &=-\tfrac{1}{\ln 2}\paren{\tfrac{2}{3}\ln 2-\ln 3}=0.9184\\
   g(z)&=2\paren{1-max\brac{z_0,z_1}}=2\paren{1-\tfrac{2}{3}}=\tfrac{2}{3}=0.667\\
   s(z)&=\tfrac{1}{2\sin (\pi /2)}\sqbrac{\sin\paren{\tfrac{2}{3}\pi}+\sin\paren{\tfrac{\pi}{3}}}=\tfrac{1}{2}=0.500\hskip 5pc\square
\end{align*}
\end{example}

For $A\in\rmob (H)$, then $\rho\in\sscript (H)$ is an $A$-\textit{maximal uncertainty state} if $\trace (\rho A_i)=\tfrac{1}{d}$,
$i=0,1,\ldots ,d-1$. Notice that $\rho$ is an $A$-maximal uncertainty state if and only if $U_{(f,A)}(\rho )=1$ for every uncertainty function
$f\colon\pscript _d\to\sqbrac{0,1}$.

\begin{thm}    
\label{thm42}
If $A=\rmob (H)$ is an atomic, projective observable $A=\brac{\ket{0}\bra{0},\ket{1}\bra{1},\ldots ,\ket{d-1}\bra{d-1}}$, then
$\rho\in\sscript (H)$ is an $A$-maximal uncertainty state if and only if $\rho =\tfrac{1}{d}I+T$ where $T\in\lscript (H)$ satisfies
$\bra{r}T\ket{r}=0$ for all $r\in\brac{0,1,\ldots ,d-1}$. 
\end{thm}
\begin{proof}
If $\rho =\tfrac{1}{d}I+T$ where $\bra{r}T\ket{r}=0$ for all $r\in\brac{0,1,\ldots ,d-1}$, then
\begin{equation*}
\trace (\rho A_i)=\tfrac{1}{d}\trace (A_i)+\trace (TA_i)=\tfrac{1}{d}+\bra{i}T\ket{i}=\tfrac{1}{d}
\end{equation*}
Conversely, suppose $\trace (\rho A_i)=\tfrac{1}{d}$, $i\in\brac{0,1,\ldots ,d-1}$. We can write
$\rho =\sum\limits _{i,j}\lambda _{ij}\ket{i}\bra{j}$ and we obtain
\begin{align*}
\tfrac{1}{d}&=\trace (\rho A_i)=\bra{i}\rho\ket{i}=\bra{i}\sum _{r,s}\lambda _{r,s}\ket{r}\bra{s}\ket{i}\\
&=\bra{i}\sum _r\lambda _{ri}\bra{r}=\lambda _{ii}
\end{align*}
Hence,
\begin{equation*}
\rho =\sum _i\lambda _{ii}\ket{i}\bra{i}+\sum _{\substack{i,j\\i\ne j}}\lambda _{ij}\ket{i}\bra{j}
   =\tfrac{1}{d}I+\sum _{\substack{i,j\\i\ne j}}\lambda _{ij}\ket{i}\bra{j}
\end{equation*}
Letting $T=\sum\limits _{\substack{i,j\\i\ne j}}\lambda _{ij}\ket{i}\bra{j}$ we conclude that
\begin{equation*}
\bra{r}T\ket{r}=\sum _{\substack{i,j\\i\ne j}}\lambda _{ij}\elbows{r\mid i}\elbows{j\mid r}=0\qedhere
\end{equation*}
\end{proof}

In general, the operator $T$ is not unique. However, if $\rho$ is diagonal relative to the basis $\brac{\ket{i}, i=0,1,\ldots ,d-1}$, we obtain
$T=0$

\begin{cor}    
\label{cor43}
If $A$ is the observable of Theorem~\ref{thm42} and $\rho$ is an $A$-maximal uncertainty state that is diagonal relative to
$\brac{\ket{i}\colon i=0,1,\ldots ,d-1}$, then $\rho =\tfrac{1}{d}I$.
\end{cor}
\begin{proof}
We have that
\begin{equation*}
\tfrac{1}{d}=\trace (\rho A_i)=\bra{i}\sum _j\lambda _j\ket{j}\elbows{j\mid i}=\lambda _i
\end{equation*}
Hence, $\rho =\tfrac{1}{d}I$.
\end{proof}

\begin{example} 
Let $A=\brac{A_0,A_1}$ be the observable of Example~1 where
$A_0=\tfrac{1}{2}\left[\begin{smallmatrix}\smallskip 1&1\\1&1\end{smallmatrix}\right]$. Suppose $\rho =\ket{\phi}\bra{\phi}$ where
$\phi =\left[\begin{smallmatrix}\smallskip \alpha\\\beta\end{smallmatrix}\right]$ is an $A$-maximal uncertainty state and we can assume
$\alpha\in\real$. Then $\alpha ^2+\ab{\beta}^2=1$ and $\tfrac{1}{2}=\trace (\rho A_0)=\elbows{\phi ,A_0\phi}$. Since
\begin{equation*}
A_0(\phi )=\tfrac{1}{2}\begin{bmatrix}\smallskip 1&1\\1&1\end{bmatrix}\begin{bmatrix}\smallskip \alpha\\\beta\end{bmatrix}
   =\tfrac{1}{2}\begin{bmatrix}\smallskip \alpha +\beta\\\alpha +\beta\end{bmatrix},
\end{equation*}
we obtain
\begin{align*}
\tfrac{1}{2}&=\elbows{\phi ,A_0\phi}
   =\tfrac{1}{2}\elbows{\begin{bmatrix}\smallskip \alpha\\\beta\end{bmatrix},
   \begin{bmatrix}\smallskip \alpha +\beta\\\alpha +\beta\end{bmatrix}}
   =\tfrac{1}{2}\sqbrac{\alpha ^2+\ab{\beta}^2+\alpha (\beta +\overline{\beta}\,)}\\
   &=\tfrac{1}{2}\sqbrac{1+2\alpha\rmre (\beta )}
\end{align*}
Hence $\alpha =0$ or $\rmre (\beta )=0$. We conclude that
\begin{equation}                
\label{eq41}
\phi =\begin{bmatrix}\smallskip \alpha\\i\sqrt{1-\alpha ^2}\ \end{bmatrix} 
\end{equation}
Thus, there are infinitely many $A$-maximal uncertainty states of form \eqref{eq41}. For example,
\begin{equation*}
\hskip 10pc
\tfrac{1}{2}\begin{bmatrix}\smallskip 1\\i\sqrt{3}\ \end{bmatrix} \hbox{ or }
\tfrac{1}{3}\begin{bmatrix}\smallskip 1\\i\sqrt{8}\ \end{bmatrix}\hskip 10pc\square
\end{equation*}
\end{example}

Just as we were able to form affine combinations of uncertainty functions, we can form affine combinations of uncertainty measures.
If $B_i=\brac{A_{ij}\colon j=0,1,\ldots ,d-1}$, $i=1,2,\ldots ,n$, where $\sum\limits _{j=0}^{d-1}A_{ij}=I$ for all $i$, are observables and
$0<\lambda _i\le 1$, $i=1,2,\ldots ,n$, with $\sum\limits _{i=1}^n\lambda _i=1$, then
\begin{equation*}
C=\brac{\sum _{i=1}^n\lambda _iA_{ij}\colon j=0,1,\ldots ,d-1}=\sum _{i-1}^n\lambda _iB_i
\end{equation*}
is an observable. We can then define the $\paren{f,\sum\lambda _iB_i}$ uncertainty measure
\begin{equation*}
U_{\paren{f,\sum\lambda _iB_i}}(\rho )=f\paren{P_\rho ^{\sum\lambda _iB_i}}=f(P_\rho ^C)=\sum\lambda _if(P_\rho ^{B_i})
\end{equation*}
We can also form the $\paren{\sum\lambda _if_i, \sum\mu_jB_j}$-uncertainty measure
\begin{equation*}
U_{\paren{\sum\lambda _if_i,\sum\mu _jB_j}}(\rho )=\sum _i\lambda _iU_{\paren{f_i,\sum\mu _jB_j}}(\rho )
   =\sum _{i,j}\lambda _i\mu _jU_{(f_i,B_j)}(\rho )
\end{equation*}


\begin{thebibliography}{99}  
\bibitem{cvtw17}P.~Coles, M.~Berta, M.~Tomamichel and S.~Wehrer, Entropic uncertain relations and their applications,
\textit{Rev. Mod. Phys.} \textbf{89}, 015002 (2017).
\bibitem{deu83}D.~Deutsch, Uncertainty in quantum measurements,
\textit{Phys. Rev, Lett,} \textbf{50}, 631 (1983).
\bibitem{fgg13}S.~Friedland, V.~Gheorghiu and G.~Gour, Universal uncertainty relations,
\textit{Phys. Rev, Lett,} \textbf{111}, 230401 (2013).
\bibitem{glm11}V.~Giovannetti, S.~Lloyd and L.~Maccone, Advances in quantum metrology,
\textit{Nat. Photon,} \textbf{5}, 222 (2011).
\bibitem{hz12}T.~Heinosaari and M.~Ziman, The Mathematical Language of Quantum Theory, Cambridge University Press, Cambridge,
2012.
\bibitem{hei72}W..~Heisenberg, \"Uber den anschaulichen Inhalt der quantentheoretischen Kinematik and Mechanik
\textit{Zeitschrift f\"ur Physic}, \textbf{43}, 172 (1972).
\bibitem{luo05}S.~Luo, Quantum versus classical uncertainty, \textit{Theor. Math. Phys.}, \textbf{143}, 681 (2005).
\bibitem{ls17}S.~Luo and Y.~Sun, Quantum coherence versus quantum uncertainty, \textit{Phys. Rev. A}, \textbf{96}, 022130 (2017).
\bibitem{mu88}H.~Maassen and J.~Uffink, Generalized entropic uncertainty relations, \textit{Phys. Rev. Lett.} \textbf{60}, 1103 (1988).
\bibitem{nc00} M.~Nielsen and J.~Chuang, \textit{Quantum Computation and Quantum Information},
Cambridge University Press, Cambridge, 2000.
\bibitem{rob29}H.~Robertson, The uncertainty principle, \textit{Phys.Rev.} \textbf{34}, 163 (1929).
\bibitem{sbwch13}J.~Schneeloch, C.~Broadbent, S.~Walborn, E.~Cavalcanti and J.~Howell,
Einstein-Podolsky-Rosen steering inequalities from entropic uncertainty relations, \textit{Phys. Rev. A}, \textbf{87}, 062103 (2013).
\bibitem{tr11}M.~Tomamichel and R.~Renner, Uncertainty relation for smooth entropies, \textit{Phys. Rev. Lett.}, \textbf{106}, 110506 (2011).
\bibitem{zt24}M.~Zhao and Y.~Tao, The uncertainty of quantum states with respect to the projective measurement, arXiv: quant-ph 2405.19588 v1 (2024).
\end{thebibliography}
\end{document}